\documentclass[letterpaper, 11pt]{article}
\usepackage{appendix}
 
\newcommand{\comment}[1]{}
 
\usepackage{graphicx}

\usepackage{subfigure}
\usepackage{amsmath}
\usepackage{amssymb}
\usepackage{mdwlist}

\usepackage{algorithm}

\usepackage{hyperref}

\usepackage{xspace}
\usepackage{setspace}

\newcommand{\sv}{V}

\newenvironment{proof}{\noindent {\bf Proof:}~}{\hspace*{\fill}\(\Box\)}

\newtheorem{theorem}{Theorem}

\newtheorem{claim}{Claim}

\newtheorem{definition}{Definition}

\newtheorem{lemma}{Lemma}

\def\noflash#1{\setbox0=\hbox{#1}\hbox to 1\wd0{\hfill}}

\newcommand{\vectorv}{{\normalfont\textbf{v}}}
\newcommand{\matrixm}{\textbf{M}}

\newcommand{\D}{{\textbf d_H}}

\newcommand{\R}{{\mbox{\it Verified\,}}}

\newcommand{\HH}{{\mathcal H}}
\newcommand{\VV}{{\mathcal V}}

\newcommand{\RBRecv}{{\tt RBRecv}}
\newcommand{\RBSend}{{\tt RBSend}}

\newcommand{\bfA}{{\bf A}}

\newcommand{\bfM}{{\bf M}}

\begin{document}
\title{Byzantine Convex Consensus: Preliminary Version\footnote{\normalsize We present an optimal algorithm in our follow-up work \cite{Tseng_BCC_optimal}.}
\footnote{\normalsize This research is supported in part by National Science Foundation award CNS 1059540. Any opinions, findings, and conclusions or recommendations expressed here are those of the authors and do not necessarily reflect the views of the funding agencies or the U.S. government.}}

\author{Lewis Tseng$^{1,3}$, and Nitin Vaidya$^{2,3}$\\~\\
 \normalsize $^1$ Department of Computer Science,\\
 \normalsize $^2$ Department of Electrical and Computer Engineering, 
 and\\ \normalsize $^3$ Coordinated Science Laboratory\\ \normalsize University of Illinois at Urbana-Champaign\\ \normalsize Email: \{ltseng3, 
nhv\}@illinois.edu \\~\\ \normalsize Technical Report} 

\date{July 3rd, 2013}
\maketitle

\begin{abstract}
{\normalsize

Much of the past work on asynchronous approximate Byzantine consensus has
assumed {\em scalar} inputs at the nodes \cite{AA_Dolev_1986, AA_nancy}.
Recent work has yielded approximate Byzantine consensus algorithms
for the case when the input at each node is a $d$-dimensional vector,
and the nodes must reach consensus on
a vector in the convex hull
of the input vectors at the fault-free nodes
\cite{herlihy_multi-dimension_AA,Vaidya_BVC}.
The $d$-dimensional vectors can be equivalently viewed as {\em points}
in the $d$-dimensional Euclidean space.
Thus, the algorithms in \cite{herlihy_multi-dimension_AA, Vaidya_BVC}
require the fault-free nodes to decide on a point
in the $d$-dimensional space.

~

In this paper, we generalize the problem to allow the decision to be
a {\em convex polytope} in the $d$-dimensional space, such that the decided
polytope is within the convex hull of the input vectors at the fault-free nodes.
We name this problem as {\em Byzantine convex consensus} (BCC), and present an
asynchronous approximate BCC algorithm with optimal fault tolerance.
Ideally, the goal here is to agree on a convex polytope that
is as large as possible.
While we do not claim that our algorithm satisfies this goal,
we show a bound on the output convex polytope chosen by our algorithm.

}
\end{abstract}

\thispagestyle{empty}
\newpage
\setcounter{page}{1}

\section{Introduction}
\label{s_intro}

Much of the past work on asynchronous approximate Byzantine consensus has
assumed {\em scalar} inputs at the nodes \cite{AA_Dolev_1986, AA_nancy}.
Recent work has yielded approximate Byzantine consensus algorithms
for the case when the input at each node is a $d$-dimensional vector,
and the nodes must reach consensus on
a vector in the convex hull
of the input vectors at the fault-free nodes
\cite{herlihy_multi-dimension_AA,Vaidya_BVC}.
The $d$-dimensional vectors can be equivalently viewed as {\em points}
in the $d$-dimensional Euclidean space.
Thus, the algorithms in \cite{herlihy_multi-dimension_AA, Vaidya_BVC}
require the fault-free nodes to decide on a point
in the $d$-dimensional space.
In this paper, we generalize the problem to allow the decision to be
a {\em convex polytope} in the $d$-dimensional space, such that the decided
polytope is within the convex hull of the input vectors at the fault-free nodes.
We name this problem as {\em Byzantine convex consensus} (BCC), and present an
asynchronous approximate BCC algorithm with optimal fault tolerance.

We consider Byzantine convex consensus (BCC) in an {\em asynchronous} system consisting of $n$ nodes,
of which at most $f$ may be Byzantine faulty. The Byzantine faulty nodes may behave
in an arbitrary fashion,
and may collude with each other. Each node $i$ has a $d$-dimensional vector of reals as its {\em input} $x_i$.
All nodes can communicate with each other directly on reliable and FIFO (first-in first-out)
channels.  Thus, the underlying communication graph can be modeled as a {\em complete graph}, 
with the set of nodes being $V=\{1,2,\cdots, n\}$.
The impossibility of {\em exact} consensus in asynchronous systems \cite{FLP_one_crash}
applies to BCC as well. Therefore, we consider 
the {\em Approximate BCC} problem with the following requirements:

\begin{itemize}

\item \textbf{Validity}: The {\em output} (or {\em decision}) at each fault-free node must be
a convex polytope in the convex hull of the $d$-dimensional input vectors at the fault-free nodes.
(In a degenerate case, the output polytope may simply be a single {\em point}.)

\item \textbf{Approximate Agreement}: For any $\epsilon > 0$, the {\em Hausdorff distance} (defined below) between the output polytopes at any two fault-free nodes must be at most $\epsilon$.


\end{itemize}
Ideally, the fault-free nodes should reach consensus on the largest
possible convex polytope that satisfies the validity constraint. We present an optimal algorithm that agrees on a convex polytope that is as {\em large} as possible under adversarial conditions in our follow-up work \cite{Tseng_BCC_optimal}.
The motivation behind reaching consensus on a convex polytope is that
a solution to BCC is expected to also facilitate solutions to a large
range of consensus problems (e.g., Byzantine vector consensus \cite{herlihy_multi-dimension_AA,Vaidya_BVC}, or convex function optimization
over a convex hull of the inputs at fault-free nodes).
Future work will explore these potential applications.

To simplify the presentation,
we do not include a {\em termination} condition above.
We instead prove that approximate agreement condition is {\em eventually} satisfied
by our proposed algorithm (in addition to validity).
However, we can augment the proposed algorithm 
 using techniques similar to those in \cite{abraham_04_3t+1_async,herlihy_multi-dimension_AA,Vaidya_BVC},
to terminate within a finite number of rounds.

\begin{definition}
\label{def:dist}
For two convex polytopes $h_1, h_2$, the {\em Hausdorff distance} is defined as \cite{Hausdorff}
\[
\D(h_1, h_2) ~~=~~ \max~~ \{~~ \max_{p_1 \in h_1}~\min_{p_2 \in h_2} d(p_1, p_2),~~~~ \max_{p_2 \in h_2}~\min_{p_1 \in h_1} d(p_1, p_2)~~\}
\]
\end{definition}
where $d(p,q)$ is the Euclidean distance between points $p$ and $q$.

\paragraph{Lower Bound on $n$:}
As noted above,
\cite{herlihy_multi-dimension_AA, Vaidya_BVC} consider the problem of
reaching approximate Byzantine consensus on a
vector (or a point) in the convex hull of the $d$-dimensional
input vectors at the fault-free nodes, and show that
$n \geq (d+2)f+1$ is necessary.
\cite{herlihy_colorless_async} generalizes
the same lower bound to colorless tasks.
The lower bound proof in
\cite{herlihy_multi-dimension_AA, Vaidya_BVC} also implies that
$n \geq (d+2)f+1$ is necessary to ensure that BCC is solvable.
We do not reproduce the lower bound proof here, but
in the rest of the paper, we assume that $n\geq (d+2)f+1$,
and also that $n\geq 2$ (because consensus is trivial when $n=1$).

\section{Preliminaries}
\label{s_ops}

Some notations introduced throughout the paper are summarized in Appendix \ref{app_s_notations}.
In this section,
we introduce operations $\HH$, $H_l$, $H$,
and a {\em reliable broadcast} primitive used later in the paper.
\begin{definition}
\label{def:hh}
Given a set of points $X$, $\HH(X)$ is defined
as the convex hull of the points in $X$.
\end{definition}
\begin{definition}
\label{def:linear_hull}
Suppose that $\nu$ convex polytopes $h_1, h_2,\cdots, h_\nu$, and $\nu$ constants $c_1, c_2, \cdots, c_\nu$
are given such that
(i) $0 \leq c_i \leq 1$ and $\sum_{i = 1}^\nu c_i = 1$,
and (ii) for $1\leq i\leq\nu$, if $c_i\neq 0$, then $h_i\neq\emptyset$.
Linear combination of these convex
polytopes, $H_l(h_1, h_2,\cdots, h_\nu;~c_1, c_2, \cdots, c_\nu)$,
 is defined as follows:
\begin{itemize}
\item Let $Q := \{ i ~|~ c_i\neq 0, ~ 1\leq i\leq \nu \}$.
\item 
$p \in H_l(h_1, h_2,\cdots, h_\nu;~c_1, c_2,\cdots, c_\nu)$ if and only if
\begin{equation}
\label{eq:linear_hull}
\text{for each~} i\in Q,
\text{~there exists~} p_i \in h_i, ~~\text{such that}~~p = \sum_{i \in Q} c_i p_i
\end{equation}
\end{itemize}
\end{definition}
Note that a convex polytope may possibly consist of a single point.
Because $h_i$'s above are all convex, $H_l(h_1, h_2,\cdots, h_\nu;~c_1, c_2,\cdots, c_\nu)$ is also a
convex polytope (proof included in Appendix \ref{a:linear_convex} for completeness).
The parameters for $H_l$ consist of two lists, a list of polytopes
$h_1,\cdots,h_\nu$, and a list of weights $c_1,\dots,c_\nu$.  With an abuse of notation, we 
will specify one or both of these lists as either a {\em row vector}
or a {\em multiset}, with the understanding that the {\em row vector} or
{\em multiset} here represents an ordered list of its elements.

Function $H$ below is called in our algorithm with parameters $(\VV,t)$ wherein
$t$ is a round index ($t\geq 0$) and $\VV$ is a set of tuples of the
form $(h,j,t-1)$, where $j$ is a node identifier; when $t=0$, $h$ is a point
in the $d$-dimensional Euclidean space, and when $t>0$, $h$
is a convex polytope.

\vspace*{2pt}

\noindent
\hrule
{\bf Function $H(\VV,t)$}
\begin{enumerate}
\item Define multiset $X:=\{h~|~ (h,j,t-1)\in \VV\}$.
In our use of function $H$, each $h\in X$ is always non-empty.

\item If $t=0$ then
${\tt temp}~:= ~ \cap_{\,C \subseteq X, |C| = |X| - f}~~\HH(C)$.\\
In our use of function $H$, when $t=0$, each $h\in X$ is simply a point.
The intersection above is over the convex hulls of the subsets of $X$ of size $|X| - f$.
\item If $t>0$ then ${\tt temp}~:=~ H_l(X; \frac{1}{|X|}, \cdots,\frac{1}{|X|})$.
~~Note that all the weights here are equal to $\frac{1}{|X|}$.
\item Return {\tt temp}. 
\end{enumerate}
\hrule

\paragraph{{\em Reliable Broadcast}\, Primitive:}

We will use the {\em reliable broadcast} primitive
from \cite{abraham_04_3t+1_async}, which is also used in other related work \cite{herlihy_multi-dimension_AA, Vaidya_BVC}.
\begin{itemize}
\item
As seen later, our algorithm proceeds in asynchronous rounds, and the
nodes perform reliable broadcast of messages that each consist of
a 3-tuple of the form $(v,i,t)$: here $i$ denotes the sender node's identifier, $t$
is round index, and $v$ is message value (the value $v$ itself is often a tuple).
The operation $\RBSend(v,i,t)$ is used by node $i$ to perform {\em reliable broadcast} of $(v,i,t)$  in round $t$.
\item
When message $(v,i,t)$ is {\em reliably received} by some node $j$,
the event $\RBRecv(v,i,t)$ is said to have occurred at node $j$
 (note that $j$ may possibly be equal to $i$).
The second element in a reliably received message 3-tuple,
namely $i$ above, is always identical to the identifier
of the node that performed the corresponding reliable broadcast.
When we say that node $j$ reliably receives $(v,i,t)$ we mean that event $\RBRecv(v,i,t)$
occurs at node $j$.
\end{itemize}
Each fault-free node performs one reliable broadcast ($\RBSend$) in each round of our algorithm.
The reliable broadcast primitive has the properties listed below,
as proved previously \cite{abraham_04_3t+1_async,herlihy_colorless_async}.
\begin{itemize}
\item {\bf Fault-Free Integrity:} If a fault-free node $i$ {\em never} reliably broadcasts $(v, i, t)$,
then no fault-free node ever reliably receives $(v,i,t)$.

\item {\bf Fault-Free Liveness:} If a fault-free node $i$ performs reliable broadcast of $(v, i, t)$, then each fault-free node eventually reliably receives $(v, i, t)$.

\item {\bf Global Uniqueness:} If two fault-free nodes $i, j$ reliably receive $(v, k, t)$ and $(w, k, t)$, respectively, then $v = w$, even if node $k$ is faulty.

\item {\bf Global Liveness:} For any two fault-free nodes $i, j$, if $i$ reliably receives $(v, k, t)$, then $j$ will eventually reliably receive $(v, k, t)$, even if node $k$ is faulty.

\end{itemize}

\section{Proposed Algorithm: Verified Averaging}
\label{s_alg}

The proposed algorithm (named {\em Verified Averaging}) proceeds in asynchronous rounds.
The input at each node $i$ is a $d$-dimensional vector of reals, denoted as $x_i$.
 In each round $t~(t \geq 0)$, each node $i$ computes a state variable $h_i$, which represents a convex polytope in the $d$-dimensional Euclidean space. We will refer to the value of $h_i$ at the {\em end} of the $t$-th round performed by node $i$ as $h_i[t]$, $t\geq 0$. Thus, for $t\geq 1$, $h_i[t-1]$ is the value of $h_i$ at the {\em start} of the $t$-th round at node $i$.

Motivated by previous work that uses a mechanism to simulate omission failures in
presence of 
Byzantine faults \cite{Welch_textbook}, our algorithm uses a similar technique,
named {\em verification}. Informally,
the {\em verification} mechanism ensures that 
if a faulty node deviates from the algorithm specification
(except possibly choosing an invalid input vector),
then its incorrect messages will be ignored by
the fault-free nodes.
Thus, aside from choosing a bad input, a faulty node cannot cause any other damage
to the execution.
Before we present the algorithm, we introduce a convention for the brevity of presentation:
\begin{itemize}
\item When we say that $(*,i,t) \in \VV$, we mean that {\em there exists $z$ such that $(z,i,t)\in\VV$}.

\item When we say that $(*,i,t) \not\in \VV$, we mean that {\em $\forall z$, $(z,i,t)\notin\VV$}.
\end{itemize}
The proposed {\em Verified Averaging} algorithm for node $i\in V$ is presented below.
All references to line numbers in our discussion refer to numbers listed
on the right side of the algorithm pseudo-code.
Whenever a message is reliably received by any node, a handler is called
to process that message. Handlers for multiple reliably received messages
may execute {\em concurrently} at a given node. For correct behavior, lines 3-7 below
are executed {\bf atomically}, and similarly, lines 11-16 are executed
{\bf atomically}.

In the proposed algorithm, in round 0, each node $i$ uses $\RBSend$ to reliably broadcast
$(x_i,i,0)$ where $x_i$ is its
input (line 1). 

Lines 2-7 specify the event handler for event $\RBRecv(x,j,0)$ at node $i$.
Whenever a new message of the form $(x,j,0)$ is reliably received by node $i$
(line 2), 
the set $\R_i[0]$ is updated (line 3).
Note that the message received by node $i$ on line 2 may possibly have been reliably
broadcast by node $i$ itself.
When size of set 
$\R_i[0]$ becomes at least $n-f$ for the first time (line 4),
node $i$ computes $h_i[0]$ (line 6); the $\R$ set used for
computing $h_i[0]$ is saved as $\R^c_i[0]$ (line 5).
Having computed $h_i[0]$, node $i$ can proceed to round 1 (line 7).
Note that new messages may still be added to $\R_i[0]$ afterwards
whenever event of the form $\RBRecv(x_j,j,0)$ occurs.
Thus, $\R_i[0]$ may continue to grow even after node $i$ has proceeded
to round 1; however, $\R^c_i[0]$ is not modified again.
Recall that lines 3-7 are performed atomically.

On entering
round $t$, $t\geq 1$, each node $i$ reliably broadcasts
$((h_i[t-1],\R^c_i[t-1]),i,t)$ (line 8).

Lines 9-16 specify the event handler for event $\RBRecv((h,\VV),j,t)$ at node $i$.
Whenever a message of the form $((h,\VV),j,t)$ is reliably received from node $j$ (line 10),
node $i$ first waits until its own set $\R_i[t-1]$ becomes
large enough to contain $\VV$. Note that $\R_i[t-1]$ is initially computed
in round $t-1$, but it may continue to grow even after node $i$ proceeds
to round $t$.
If the condition $\VV\subseteq\R_i[t-1]$ never becomes true, then this
message is not processed further.
The message $((h,\VV),j,t)$ is considered correct if all the following
conditions are true: (i) $\VV\subseteq \R_i[t-1]$,
(ii) $|\VV|\geq n-f$,
(iii) $(*,j, t-2)\in \VV$, and
(iv) $h=H(\VV,t-1)$.
Conditions (ii), (iii) and (iv) are tested in line 11,
and node $i$ does not reach line 11 until condition (i)
becomes true at line 10.
If the message $((h,\VV),j,t)$ is considered correct, then $(h,j,t-1)$ it is added to $\R_i[t]$ (line 12).

When both conditions $|\R_i[t]|\geq n-f$ and $(h_i[t-1],i,t-1)\in \R_i[t]$
are true for the first time (line 13),
node $i$ computes $h_i[t]$ (lines 14 and 15), and then proceeds to round $t+1$.
 Similar to round 0, the set used
in computing $h_i[t]$ is saved as $\R^c_i[t]$ (line 14).  
While $\R^c[t]$ remains unchanged afterwards, $\R_i[t]$ may continue to grow, even
after node $i$ proceeds to round $t+1$, if new round $t$ messages are reliably received
later.
Recall that lines 11-16 are performed atomically.

\vspace*{8pt}
\hrule

\vspace*{2pt}

\noindent {\bf Verified Averaging Algorithm:} Steps at node $i$

\vspace*{4pt}

\hrule

\vspace*{8pt}

\noindent {\bf Initialization:} All sets used below are initialized to $\emptyset$.\\

\noindent {\bf Round 0:} 	

\begin{itemize}
\item $\RBSend(x_i,i,0)$ 		\hfill 1

\item
{\bf Event handler for event $\RBRecv(x, j, 0)$ at node $i$} : \hfill 2\\

\centerline{\underline{\em Lines 3-7 are performed atomically.}}

\begin{list}{}{}
\item[\hspace*{0.4in}$-$] $\R_i[0] := \R_i[0] \cup \{(x,j,-1)\}$ \hfill 3 \\

\indent {\em Comment}: The third element of the 3-tuple added to $\R$ above is set\\
 \indent as $-1$ to facilitate consistent treatment in round 0 and rounds $t>0$. \\

\item[\hspace*{0.4in}$-$] When $|\R_i[0]| \geq n-f$ \& $(x_i,i,-1) \in \R_i[0]$  \underline{both true for the first time},	\hfill 4 

	\hspace*{1in} $\R^c_i[0]:= \R_i[0]$ \hfill 5

	\hspace*{1in} $h_i[0] := H(\R^c_i[0],0)$ \hfill 6

	\hspace*{1in} Proceed to round 1  \hfill 7 
\end{list}

\end{itemize}

\noindent {\bf Round $t ~ (t \geq 1)$:} 

\begin{itemize}
\item $\RBSend((h_i[t-1], \R^c_i[t-1]), i, t)$ \hfill 8

\item
{\bf Event handler for event $\RBRecv((h,\VV), j, t)$ at node $i$} : \hfill 9

\begin{itemize}

\item Wait until $\VV \subseteq \R_i[t-1]$ \hfill 10 \\

\centerline{\underline{\em Lines 11-16 are performed atomically.}}

\item
If $|\VV|\geq n-f$ and $(*,j,t-2)\in \VV$ and $h = H(\VV,t-1)$  \hfill 11 \\
then $\R_i[t] := \R_i[t] \cup \{(h,j,t-1)\}$ \hfill 12 \\ 

\item
When $|\R_i[t]| \geq n-f$ \& $(h_i[t-1],i,t-1) \in \R_i[t]$ \underline{both true for first time} \hfill 13

	\hspace*{1.6in} $\R^c_i[t]:= \R_i[t]$ \hfill 14

	\hspace*{1.6in} $h_i[t] := H(\R^c_i[t],t)$ \hfill 15

	\hspace*{1.6in} Proceed to round $t+1$  \hfill 16 

\end{itemize}
\end{itemize}

\hrule

\vspace*{8pt}

~



\begin{definition}
\label{d_verified}
A node $k$'s execution of round $r$, $r\geq 0$, is said to be
\underline{\bf verified by a fault-free node $i$} if,
eventually node $i$ reliably receives message
of the form $((h,\VV),k,r+1)$ from node $k$, and subsequently adds $(h,k,r)$ to $\R_i[r+1]$. 
Note that node $k$ may possibly be faulty.
Node $k$'s execution of round $r$ is said to be \underline{\bf verified} if it
is verified by at least one fault-free node.
\end{definition}
We now introduce some more notations (which are also summarized in Appendix \ref{app_s_notations}):
\begin{itemize}
\item For a {\em given} execution of the proposed {\em Verified Averaging} algorithm, let $F$ denote the {\em actual} set of faulty nodes in the execution.
Let $|F|=\phi$. Thus, $0\leq \phi\leq f$.
\item For $r\geq 0$, let $F_v[r]$ denote the set of faulty nodes whose round $r$ execution is verified by at least one fault-free node, as per Definition \ref{d_verified}.
Note that $F_v[r]\subseteq F$.
\item Define $\overline{F_v}[r]=F-F_v[r]$, for $r\geq 0$.
\end{itemize}
For each faulty node $k\in F_v[r]$, by
	Definition \ref{d_verified}, there must exist
	a fault-free node $i$ that eventually reliably receives
	a message of the form $((h,\VV),k,$r+1$)$ from
	node $k$, and adds $(h,k,r)$ to $\R_i[r+1]$.
	Given these $h$ and $\VV$, for future reference, let us define 
	\begin{eqnarray}
	 h_k[r]&=&h \label{e_faulty_h} \\
	 \R^c_k[r]&=&\VV \label{e_faulty_V}
	\end{eqnarray}
	Node $i$ verifies node $k$'s round $r$ execution after node $i$
	has entered its round $r+1$.
	Since round $r$ execution of faulty node $k$ above is verified by fault-free
	node $i$, due to the check performed by node $i$ at line 11, the equality below holds
	for $h_k[r]$ and $\R^c_k[r]$ defined in (\ref{e_faulty_h}) and (\ref{e_faulty_V}).
	\begin{eqnarray}
	h_k[r]& = & H(\R^c_k[r],r) \label{e_faulty_H}
	\end{eqnarray}
(The proof of Claim \ref{c_ver} in Appendix \ref{a_claims} elaborates on the
above equality.)
	While the algorithm requires each node $k$ to maintain
	variables $h_k[r]$ and $\R^c_k[r]$, we cannot assume correct behavior on
	the part of the faulty nodes. However, from the perspective of
	each fault-free node that verifies the round $r$ execution of faulty node $k\in F_v[r]$,
	node $k$ behaves ``as if'' these local variable take the values specified
	in (\ref{e_faulty_h}) and (\ref{e_faulty_V}) that satisfy (\ref{e_faulty_H}). 
	Note that if the round $r$ execution of a faulty node $k$ is verified by more than
	one fault-free node, due to the {\em Global Uniqueness} of reliable broadcast,
	all these fault-free nodes must have reliably received identical round $r+1$ messages
	from node $k$.  

Proofs of Lemmas \ref{l_progress}, \ref{lemma:J_in_H0} and \ref{lemma:always_notFv_not_verified}
below
are presented in Appendices \ref{a_lemma_progress}, \ref{a_lemma_J}, and
\ref{a_always}, respectively.
These lemmas are used to prove the correctness of the {\em Verified Averaging} algorithm.


\begin{lemma}
\label{l_progress}
If all the fault-free nodes progress to the start of round $t$, $t \geq 0$, then all the
fault-free nodes will eventually progress to the start of round $t+1$.
\end{lemma}

\begin{lemma}
\label{lemma:J_in_H0}
For each node $i\in V-\overline{F_v}[0]$, the polytope $h_i[0]$ is non-empty.
\end{lemma}


\begin{lemma}
\label{lemma:always_notFv_not_verified}
For $r\geq 0$,
if $b \in \overline{F_v}[r]$, then for all $\tau\geq r$,
\begin{itemize}
\item $b \in \overline{F_v}[\tau]$, and 
\item for all $i\in V-\overline{F_v}[\tau+1]$, $(*,b,\tau)\not\in\R^c_i[\tau+1]$.
\end{itemize}
\end{lemma}

\section{Correctness}
\label{s_correctness}

We first
introduce some terminology and definitions related to matrices. Then, we develop a
{\em transition matrix} representation of the proposed algorithm, and use
that to prove its correctness.

\subsection{Matrix Preliminaries}

We use boldface upper case letters to denote matrices, rows of matrices, and their elements. For instance, $\bfA$ denotes a matrix, $\bfA_i$ denotes the $i$-th row of matrix $\bfA$, and $\bfA_{ij}$ denotes the element at the intersection of the $i$-th row and the $j$-th column of matrix $\bfA$.

\begin{definition}
\label{d_stochastic}
A vector is said to be stochastic if all its elements
are non-negative, and the elements add up to 1.
A matrix is said to be row stochastic if each row of the matrix is a
stochastic vector. 
\end{definition}
For matrix products, we adopt the ``backward'' product convention below, where $a \leq b$,
\begin{equation}
\label{backward}
\Pi_{\tau=a}^b \bfA[\tau] = \bfA[b]\bfA[b-1]\cdots\bfA[a]
\end{equation}
For a row stochastic matrix $\bfA$,
 coefficients of ergodicity $\delta(\bfA)$ and $\lambda(\bfA)$ are defined as
follows \cite{Wolfowitz}:
\begin{eqnarray*}
\delta(\bfA) & = &   \max_j ~ \max_{i_1,i_2}~ \| \bfA_{i_1\,j}-\bfA_{i_2\,j} \| \label{e_zelta} \\
\lambda(\bfA) & = & 1 - \min_{i_1,i_2} \sum_j \min(\bfA_{i_1\,j} ~, \bfA_{i_2\,j}) \label{e_lambda}
\end{eqnarray*}
\begin{claim}
\label{claim_zelta}
For any $p$ square row stochastic matrices $\bfA(1),\bfA(2),\dots, \bfA(p)$, 
\begin{eqnarray*}
\delta(\Pi_{\tau=1}^p \bfA(\tau)) ~\leq ~
 \Pi_{\tau=1}^p ~ \lambda(\bfA(\tau)).
\end{eqnarray*}
\end{claim}
Claim \ref{claim_zelta} is proved in \cite{Hajnal58}.
%
%
%
%
Claim \ref{c_lambda_bound} below follows directly from the definition of $\lambda(\cdotp)$. 
\begin{claim}
\label{c_lambda_bound}
If there exists a constant $\gamma$, where $0<\gamma\leq 1$, such
that, for any pair of rows $i,j$ of matrix $\bfA$, there exists a column
$g$ (that may depend on $i,j$) such that,
$\min (\bfA_{ig},\bfA_{jg}) \geq \gamma$,
then
 $\lambda(\bfA)\leq 1-\gamma<1$.
\end{claim}

Let $\vectorv$ be a column vector with $n$ elements, such that the $i$-th element of vector $\vectorv$, namely $\vectorv_i$, is a convex polytope in the $d$-dimensional Euclidean space. Let $\bfA$ be a $n\times n$ row stochastic square matrix. Then multiplication of
matrix $\bfA$ and vector $\vectorv$ is performed by multiplying each row of $\bfA$
with column vector $\vectorv$ of polytopes. Formally, 
\begin{equation}
\label{eq:multiplication}
\bfA \vectorv = [H_l(\vectorv^T; \bfA_1)
~~~~~H_l(\vectorv^T; \bfA_2)
~~~~~...~~~~~
H_l(\vectorv^T; \bfA_n)]^T
\end{equation}
where $^T$ denotes the transpose operation (thus, $\vectorv^T$ is the transpose of $\vectorv$).
$H_l$ is defined in Definition \ref{def:linear_hull}. Thus, the result of the multiplication
$\bfA\vectorv$ is a column vector consisting of $n$ convex polytopes.
Similarly, product of row vector $\bfA_i$ and above vector $\vectorv$ is
obtained as follows, and it is a polytope.
\begin{eqnarray}
\label{e_r_c}
\bfA_i\vectorv &=& H_l(\vectorv^T\,;~\bfA_i)
\end{eqnarray}

\subsection{Transition Matrix Representation of {\em Verified Averaging}}

\noindent
Let $\vectorv[t]$, $t\geq 0$, denote a column vector of length $|V|=n$.
In the remaining discussion,
we will refer to $\vectorv[t]$ be the state of the system at the end of round $t$.
In particular,
$\vectorv_i[t]$ for $i\in V$ is viewed as
the state of node $i$ at the end of round $t$.
We define $\vectorv[0]$ as follows:
\begin{itemize}
\item[(I1)] For each fault-free node $i\in V-F$, $\vectorv_i[0]:=h_i[0]$.
\item[(I2)] For each faulty node $k\in F_v[0]$, 
	$\vectorv_k[0]:=h_k[0]$, where
	$h_k[0]$ is defined in (\ref{e_faulty_h}).
\item[(I3)] For each faulty node $k\in \overline{F_v}[0]$,
$\vectorv_k[0]$ is {\em arbitrarily} defined as the origin,
or the all-0 vector.
We will justify this arbitrary choice later.
\end{itemize}

We will show that the state evolution 
can be represented in a matrix form as in (\ref{matrix:alg1}), for a suitably
chosen $n\times n$ matrix  $\matrixm[t]$.
$\matrixm[t]$ is said to be the {\em transition matrix} for round $t$.
\begin{equation}
\label{matrix:alg1}
\vectorv[t] = \matrixm[t]~\vectorv[t-1], ~~~~~ t\geq 1
\end{equation}
For all $t\geq 0$,
Theorem \ref{t_M} below proves that,
for each $i\in V-\overline{F_v}[t]$, $h_i[t]=\vectorv_i[t]$.

Given a particular execution of the algorithm, we construct the transition matrix
$\bfM[t]$
for round $t\geq 1$ using
the following procedure.

\vspace*{10pt}
\hrule

\vspace*{2pt}

\noindent {\bf Construction of the Transition Matrix for Round $t~(t \geq 1)$}

\vspace*{4pt}

\hrule

\vspace*{8pt}

\begin{itemize}
\item For each node $i \in V-\overline{F_v}[t]$, and each $k\in V$:

\begin{list}{}{}
\item{} If $(h_k[t-1],k,t-1) \in \R^c_i[t]$, then 
\begin{equation}
\label{eq:matrix_i}
\matrixm_{ik}[t] = \frac{1}{|\R^c_i[t]|}
\end{equation}

\item{} Otherwise,
\begin{equation}
\label{eq:matrix_i-2}
\matrixm_{ik}[t] = 0
\end{equation}

\end{list}

{\em Comment:}
For a faulty node $i\in F_v[t]$, $h_i[t]$ and $\R^c_i[t]$
are defined in (\ref{e_faulty_h}) and (\ref{e_faulty_V}).\\

\item For each node $j \in \overline{F_v}[t]$, 
and each $k \in V$,
\begin{eqnarray}
\matrixm_{jk}[t] &=& \frac{1}{n}
\label{e_fv}
\end{eqnarray}

\end{itemize}

\hrule

\vspace*{8pt}

~

\begin{theorem}
\label{t_M}
For $r\geq 0$,
with state evolution
specified as $\vectorv[r+1]=\bfM[r+1]\vectorv[r]$ using $\bfM[r+1]$ constructed above,
for all $i\in V-\overline{F_v}[r]$,
(i) $h_i[r]$ is non-empty, and
 (ii) $h_i[r]=\vectorv_i[r]$.
\end{theorem}
\begin{proof}

The proof of the theorem is by induction. The theorem holds for $r=0$ due to
Lemma \ref{lemma:J_in_H0}, and the choice
of the elements of $\vectorv[0]$, as specified in (I1), (I2) and (I3) above.

Now, suppose that the theorem holds for $r=t-1$ where $t-1\geq 0$,
and prove it for $r=t$. 
Thus, by induction hypothesis, for
all $i\in V-\overline{F_v}[t-1]$, $h_i[t-1]=\vectorv_i[t-1]\neq \emptyset$.
Now, $\vectorv[t] = \bfM[t]\vectorv[t-1]$.
\begin{itemize}
\item In round $t\geq 1$, each fault-free node $i\in V-F$ computes its new state 
$h_i[t]$
at line 15 using function $H(\R^c_i[t],t)$. The function $H(\R^c_i[t],t)$ for $t\geq 1$ then computes a
linear combination of $|\R^c_i[t]|$ convex hulls, with all the weights being equal to
$\frac{1}{|\R^c_i[t]|}$.
Also,
by Definition \ref{d_verified} and the definition of $\overline{F_v}[t-1]$, if $(h,j,t-1)\in \R^c_i[t]$, then $j\not\in \overline{F_v}[t-1]$ (i.e., $j\in V-\overline{F_v}[t-1]$).
Therefore, if $(h,j,t-1)\in \R^c_i[t]$, then either $j$ is fault-free, or it is faulty and its
round $t-1$ execution is verified: thus, $h=h_j[t-1]$.

Also, by induction hypothesis, $h=h_j[t-1]\neq\emptyset$.
This implies that $h_i[t]=H(\R^c_i[t],t)$ is non-empty.

Then observe that, by defining $\bfM_{ik}[t]$ elements as in (\ref{eq:matrix_i})
and (\ref{eq:matrix_i-2}), we ensure that
$\bfM_i[t]\vectorv[t-1]$ equals $H(\R^c_i[t],t)$, and hence equals $h_i[t]$.

\item
For $i\in F_v[t]$ as well, as shown in (\ref{e_faulty_H}), $h_i[t]=H(\R^c_i[t],t)$, where
$h_i[t]$ and $\R^c_i[t]$ are as defined in (\ref{e_faulty_h})
and (\ref{e_faulty_V}).
The function $H(\R^c_i[t],t)$ for $t\geq 1$ then computes a
linear combination of $|\R^c_i[t]|$ convex hulls, with all the weights being equal to
$\frac{1}{|\R^c_i[t]|}$.
Consider an element $(h,j,t-1)$ in $\R^c_i[t]$. We argue that $j \in V - \overline{F_v}[t-1]$. Suppose this is not true, i.e., $j \in \overline{F_v}[t-1]$. 
By Definition \ref{d_verified}, node $i$'s round $t$ execution is verified by some fault-free node $k$, which implies that eventually, $\R^c_i[t] \subseteq \R_k[t]$. However, since $k$ is fault-free, and $(h,j,t-1) \not\in \R_k[t]$, a contradiction.
Hence, if $(h,j,t-1)\in \R^c_i[t]$, then $j\in V-\overline{F_v}[t-1]$.
That is, if $(h,j,t-1)\in \R^c_i[t]$, then either $j$ is fault-free, or it is faulty and its
round $t-1$ execution is verified: thus, $h=h_j[t-1]$.

Also, by induction hypothesis, $h=h_j[t-1]\neq\emptyset$.
This implies that $h_i[t]=H(\R^c_i[t],t)$ is non-empty.

Then observe that, by defining $\bfM_{ik}[t]$ elements as in (\ref{eq:matrix_i})
and (\ref{eq:matrix_i-2}), we ensure that
$\bfM_i[t]\vectorv[t-1]$ equals $H(\R^c_i[t],t)$, and hence equals $h_i[t]$. 

\end{itemize}

\end{proof}

\comment{++++++++ old text++++++
From the discussion in the above proof,
and Lemma \ref{lemma:always_notFv_not_verified}, it follows that
the state $\vectorv_j[t]$ for each node $j\in\overline{F_v}[t]$ does not affect the state
of the nodes $V-\overline{F_v}[\tau]$, for $\tau\geq t+1$. Therefore,
the value of $\vectorv_j[t]$ for $j\in \overline{F_v}[t]$ is
of no consequence to the other nodes.
++++++++++}
Now, we argue that for $t \geq 0$, the state $\vectorv_j[t]$ for each node $j\in\overline{F_v}[t]$ does not affect the state
of the nodes $V-\overline{F_v}[\tau]$, for $\tau\geq t+1$.
From the discussion in the above proof, we see that for $j\in\overline{F_v}[t]$, $(*,j,t)\not\in \R^c_i[t+1]$ for $i \in V-\overline{F_v}[t+1]$. Thus, the sate $\vectorv_j[t]$ does not affect the state $h_i[t+1]$. Then, by Lemma \ref{lemma:always_notFv_not_verified}, if $j\in\overline{F_v}[t]$, then $j\in\overline{F_v}[\tau]$, for $\tau\geq t+1$. Thus, by the same argument, the sate $\vectorv_j[\tau]$ does not affect the state $h_i[t+1]$.
This justifies the somewhat arbitrary choice of
$\vectorv_j[0]$ for $j\in\overline{F_v}[0]$,
and $\matrixm_{jk}[t]$ in (\ref{e_fv}) for $j\in \overline{F_v}[t], ~t\geq 1$.
This choice does simplify the remaining proof somewhat.

The above discussion shows that, for $t\geq 1$, the evolution of $\vectorv[t]$ can be written as
in (\ref{matrix:alg1}), that is, $\vectorv[t]=\matrixm[t]\vectorv[t-1]$.
Given the matrix product definition in (\ref{eq:multiplication}),
it is easy to verify that
\[ \bfM[\tau+1] ~ \left(\bfM[\tau] \vectorv[\tau-1]\right) ~=~ \left(\bfM[\tau+1]\bfM[\tau]\right)~\vectorv[\tau-1]
\text{~ for~} \tau \geq 1.
\]
Therefore, by repeated application of (\ref{matrix:alg1}),
we obtain:
\begin{eqnarray}
\vectorv[t] & = & \left(\,\Pi_{\tau=1}^t \matrixm[\tau]\,\right)\, \vectorv[0],
 ~~~~ t\geq 1
\label{e_unroll}
\end{eqnarray}
Recall that we adopt the ``backward'' matrix product convention presented in
(\ref{backward}).

\begin{lemma}
\label{lemma:transition_matrix}
For $t\geq 1$,
transition matrix $\bfM[t]$ constructed using the above procedure satisfies the following conditions. 
\begin{itemize}
\item For $i,j \in V$, there exists a fault-free node $g(i,j)$ such that $\bfM_{ig(i,j)}[t] \geq \frac{1}{n}$. 

\item $\bfM[t]$ is a row stochastic matrix, and $\lambda(\bfM[t])\leq 1-\frac{1}{n}$.


\end{itemize}

\end{lemma}
The proof of Lemma \ref{lemma:transition_matrix} is presented in Appendix \ref{a_t_matrix}.

\subsection{Correctness of {\em Verified Averaging}}


\begin{definition}
\label{def:valid_hull}
A convex polytope $h$ is said to be {\em valid} if every point in $h$ is in the convex hull of the inputs at the fault-free nodes. 
\end{definition}

Lemmas \ref{lemma:valid_initial_hull}
and \ref{lemma:linear_valid} below are proved in  
Appendices \ref{app_s_lemma:valid_initial_hull} and \ref{app_s_lemma:linear_valid}, respectively.

\begin{lemma}
\label{lemma:valid_initial_hull}
$h_i[0]$ for each node $i \in \sv - \overline{F_v}[0]$ is valid.
\end{lemma}

\begin{lemma}
\label{lemma:linear_valid}
Suppose non-empty convex polytopes $h_1, h_2, \cdots, h_k$ are all valid. Consider $k$ constants $c_1, c_2, \cdots, c_k$ such that $0 \leq c_i \leq 1$ and $\sum_{i = 1}^k c_i = 1$.
Then the linear combination of these convex polytopes,
$H_l(h_1, h_2, \cdots, h_k; c_1, c_2, \cdots, c_k)$, is valid.
\end{lemma}




\begin{theorem}
\label{thm:correctness}
Verified Averaging satisfies the {\em validity} and
{\em approximate agreement} properties
after a large enough number of asynchronous rounds.
\end{theorem}

\begin{proof}
Repeated applications of
Lemma \ref{l_progress} ensures that the fault-free nodes will progress
from round 0 through round $r$, for any $r\geq 0$, allowing us to use (\ref{e_unroll}). Consider round $t \geq 1$.
Let
\begin{eqnarray}\bfM^* &=& \Pi_{\tau=1}^t \bfM[\tau].
\label{e_Mstar}
\end{eqnarray}
(To simplify the presentation, we do not include the round index $[t]$
in the notation $\bfM^*$ above.)
 Then $\vectorv[t]=\bfM^* \vectorv[0]$.
By Lemma \ref{lemma:transition_matrix},
each $\bfM[t]$ is a {\em row stochastic} matrix,
therefore, $\bfM^*$ is also row stochastic. 
By Lemma \ref{lemma:valid_initial_hull},
$h_i[0]=\vectorv_i[0]$ for each $i\in \sv - \overline{F_v}[0]$ is valid. Therefore,
by Lemma \ref{lemma:linear_valid}, $\bfM_i^*\vectorv[0]$ for each $i\in V-F$
is valid. Also, by Theorem \ref{t_M} and (\ref{e_unroll}), $h_i[t]=\bfM_i^*\vectorv[0]$
for $i\in V-F$.
Thus, $h_i[t]$ is valid, and {\em Verified Averaging} satisfies the
validity condition for all $t\geq 0$.

Let us define $\alpha = 1-\frac{1}{n}$.
By Lemma \ref{lemma:transition_matrix}, $\lambda(\bfM[t]) \leq 1-\frac{1}{n}=\alpha$. Then by Claim \ref{claim_zelta},
\begin{eqnarray}
\label{e_alpha}
\delta(\bfM^*)=
\delta(\Pi_{\tau=1}^t \bfM[\tau])
~ \leq ~ \lim_{t\rightarrow\infty} \Pi_{\tau=1}^{t} \lambda(\bfM[\tau]) 
 ~\leq~ {\left(1-\frac{1}{n}\right)}^t ~=~\alpha^t
\end{eqnarray}
%
%

Consider any two fault-free nodes $i,j\in V-F$.
By (\ref{e_alpha}), $\delta(\bfM^*)\leq \alpha^t$. 
Therefore, by the definition of $\delta(\cdot)$, for $1 \leq k \leq n$,
\begin{equation}
\label{eq:delta1}
\| \bfM^*_{ik} -  \bfM^*_{jk}\| \leq \alpha^t
\end{equation}
%
%
By Lemma \ref{lemma:always_notFv_not_verified}, and construction of the transition matrices,
it should be easy to see that $\bfM^*_{ib}=0$ for $b\in \overline{F_v}[0]$.
Then, for
any point $p_i^*$ in $h_i[t]=\bfM^*_i\vectorv[0]$,
there must exist, for all $k\in V-\overline{F_v}[0]$, $p_k \in h_k[0],$ such that
\begin{equation}
\label{pi}
p_i^* = \sum_{k\in V-\overline{F_v}[0]} \bfM^*_{ik} p_k
= \left(\sum_{k\in V-\overline{F_v}[0]} \bfM^*_{ik} p_k(1),~~\sum_{k\in V-\overline{F_v}[0]} \bfM^*_{ik} p_k(2), \cdots, \sum_{k\in V-\overline{F_v}[0]} \bfM^*_{ik} p_k(d)\right)
\end{equation}
where $p_k(l)$ denotes the value of $p_k$'s $l$-th coordinate.
Now choose point $p_j^*$ in $h_j[t]$ defined as follows.
\begin{equation}
\label{pj}
p_j^* = \sum_{k\in V-\overline{F_v}[0]} \bfM^*_{jk} p_k 
= \left(\sum_{k\in V-\overline{F_v}[0]} \bfM^*_{jk} p_k(1),~~ \sum_{k\in V-\overline{F_v}[0]} \bfM^*_{jk} p_k(2), \cdots,  \sum_{k\in V-\overline{F_v}[0]} \bfM^*_{jk} p_k(d)\right)
\end{equation}

Then the Euclidean distance between $p_i^*$ and $p_j^*$ is $d(p_i^*,p_j^*)$.
The following derivation is obtained by simple algebraic manipulation, using (\ref{eq:delta1}),
(\ref{pi}) and (\ref{pj}).
The omitted steps in the algebraic manipulation are shown in Appendix \ref{a_hausdorff}.
\begin{align}
d(p_i^*, p_j^*) &=  \sqrt{\sum_{l=1}^d (p_i^*(l) - p_j^*(l))^2} = \sqrt{\sum_{l=1}^d \left(\sum_{k\in V-\overline{F_v}[0]} \bfM^*_{ik} p_k(l) - \sum_{k\in V-\overline{F_v}[0]} \bfM^*_{jk} p_k(l)\right)^2}\nonumber\\
&\leq \alpha^t \sqrt{\sum_{l=1}^d \left(\sum_{k\in V-\overline{F_v}[0]} \| p_k(l) \|\right)^2}\label{eq:d_pipj1} \leq ~~ \alpha^t\Omega
\end{align} 
where
$\Omega = \max_{p_k \in h_k[0],k\in V-\overline{F_v}[0]} \sqrt{\sum_{l=1}^d (\sum_{k\in V-\overline{F_v}[0]} \| p_k(l) \|)^2}$.
Because the $h_k[0]$'s in the definition of $\Omega$ are all valid (by
Lemma \ref{lemma:valid_initial_hull}),
$\Omega$ can itself be upper bounded by a function of the input vectors at the fault-free
nodes.
Since $\alpha=1-\frac{1}{n}<1$, for large enough $t$, $\alpha^t\Omega<\epsilon$, for any given
$\epsilon$. Therefore, for fault-free $i,j$, for large enough $t$, for each point $p_i^*\in h_i[t]$
there exists a point $p_j^*[t]\in h_j[t]$ such that $d(p^*_i,p^*_j)<\epsilon$ (and, similarly, vice-versa).
Thus, by Definition \ref{def:dist}, eventually
Hausdorff distance $\D(h_i[t],h_j[t]) <\epsilon$.
Since this holds true for any pair of fault-free nodes
$i,j$, approximate agreement property is eventually satisfied.
\end{proof}

\section{Convex Polytope Obtained by {\em Verified Averaging}}
\label{s_size}

Recall that $|F|=\phi\leq f$.
Let $G = \cup_{i \in \{\sv - F\}} x_i$ be the set of the inputs at all fault-free nodes.
Thus, $|G|=n-\phi\geq n-f$.
Define a convex polytope $I$ as follows.
\begin{equation}
\label{eq:I}
I = \cap_{D \subset G, |D| = n - 2f-\phi}\, \HH(D)
\end{equation}

%

\begin{lemma}
\label{lemma:opt}
For all $i\in V-F$ and $t\geq 0$,
 $I\subseteq h_i[t]$.
\end{lemma}


The lemma is proved in Appendix \ref{a_opt}.
The lemma establishes a ``lower bound'' on the convex polytope that the fault-free
nodes decide on.
Due to Theorem \ref{t_M}, $h_i[t]$ is always non-empty for $i\in V-F$.
However, $I$ may possibly be empty, depending on the inputs at the fault-free
nodes, and the total number of nodes.
We believe that it may be possible to improve the above ``lower bound'' by using
a somewhat more complex algorithm (using the {\em stable-vectors} primitive
from \cite{herlihy_colorless_async} in round 0).
This improvement is left as a topic for future work. In the follow-up work \cite{Tseng_BCC_optimal}, we present an optimal algorithm that agrees on a convex polytope that is as {\em large} as possible under adversarial conditions

\section{Summary}

This paper addresses {\em Byzantine Convex Consensus} (BCC), wheresin each
node has a $d$-dimensional vector as its {\em input}, and each fault-free node has decide on a
polytope that is in the {\em convex hull} of the input vectors at the fault-free nodes. 
We present an {\em asynchronous approximate} BCC algorithm with optimal fault tolerance, 
and establish a lower bound on the convex polytope agreed upon.


\begin{thebibliography}{10}

\bibitem{abraham_04_3t+1_async}
I.~Abraham, Y.~Amit, and D.~Dolev.
\newblock Optimal resilience asynchronous approximate agreement.
\newblock In {\em OPODIS}, pages 229--239, 2004.

\bibitem{Welch_textbook}
H.~Attiya and J.~Welch.
\newblock {\em Distributed Computing: Fundamentals, Simulations, and Advanced
  Topics}.
\newblock Wiley Series on Parallel and Distributed Computing, 2004.

\bibitem{AA_Dolev_1986}
D.~Dolev, N.~A. Lynch, S.~S. Pinter, E.~W. Stark, and W.~E. Weihl.
\newblock Reaching approximate agreement in the presence of faults.
\newblock {\em J. ACM}, 33:499--516, May 1986.

\bibitem{FLP_one_crash}
M.~J. Fischer, N.~A. Lynch, and M.~S. Paterson.
\newblock Impossibility of distributed consensus with one faulty process.
\newblock {\em J. ACM}, 32:374--382, April 1985.

\bibitem{Hajnal58}
J.~Hajnal.
\newblock Weak ergodicity in non-homogeneous markov chains.
\newblock In {\em Proceedings of the Cambridge Philosophical Society},
  volume~54, pages 233--246, 1958.

\bibitem{Hausdorff}
D.~Huttenlocher, G.~Klanderman, and W.~Rucklidge.
\newblock Comparing images using the Hausdorff distance.
\newblock {\em IEEE Transactions on Pattern Analysis and Machine Intelligence},
 15(9):850--863, 1993.

\bibitem{AA_nancy}
N.~A. Lynch.
\newblock {\em Distributed Algorithms}.
\newblock Morgan Kaufmann, 1996.

\bibitem{herlihy_multi-dimension_AA}
H.~Mendes and M.~Herlihy.
\newblock Multidimensional approximate agreement in byzantine asynchronous
  systems.
\newblock In {\em ACM Symposium on Theory of Computing (STOC)}, 2013.

\bibitem{herlihy_colorless_async}
H.~Mendes, C.~Tasson, and M.~Herlihy.
\newblock The topology of asynchronous byzantine colorless tasks.
\newblock {\em CoRR}, abs/1302.6224, 2013.

\bibitem{tverberg}
M.~A. Perles and M.~Sigorn.
\newblock A generalization of Tverberg's theorem.
\newblock {\em CoRR}, abs/0710.4668, 2007.

\bibitem{Tseng_BCC_optimal}
L.~Tseng and N.~H. Vaidya.
\newblock Byzantine Convex Consensus: An Optimal Algorithm.
\newblock Tech. Report, CSL UIUC, 2013.
\newblock http://www.crhc.illinois.edu/wireless/papers/BCC\_optimal\_tech.pdf.

\bibitem{Vaidya_BVC}
N.~H. Vaidya and V.~K. Garg.
\newblock Byzantine vector consensus in complete graphs.
\newblock {\em CoRR}, abs/1302.2543, 2013. To appear at ACM PODC 2013.

\bibitem{Wolfowitz}
J.~Wolfowitz.
\newblock Products of indecomposable, aperiodic, stochastic matrices.
\newblock In {\em Proceedings of the American Mathematical Society}, volume~14,
  pages 733--737, 1963.

\end{thebibliography}

\clearpage
\appendix

\section{Notations}
\label{app_s_notations}

This appendix summarizes some of the notations and terminology introduced throughout the paper.

\begin{itemize}
\item $n =$ number of nodes. We assume that $n\geq 2$.
\item $f =$ maximum number of Byzantine nodes.
\item $\sv = \{1, 2, \cdots, n\}$ is the set of all nodes.
\item $d = $ dimension of the input vector at each node.
\item $d(p, q) = $ the function returns the Euclidean distance between points $p$ and $q$.
\item $d_H(h_1, h_2) = $ the Hausdorff distance between convex polytopes $h_1, h_2$.
\item $\HH(C) = $ the convex hull of a multiset $C$.
\item $H_l(h_1, h_2, \cdots, h_k;~c_1, c_2, \cdots, c_k)$, defined in Section \ref{s_ops},
is a linear combination of convex polytopes $h_1, h_2, ..., h_k$ with weights $c_1,c_2,\cdots,c_k$.
\item $H(\VV,t)$ is a function defined in Section \ref{s_ops}.
\item $| X | = $ the size of a {\em multiset} or {\em set} $X$.
\item $\| a \| = $ the absolute value of a real number $a$.
\item $F$ denotes the {\em actual} set of faulty nodes in an execution of the algorithm.
\item $\phi=|F|$. Thus, $0\leq \phi\leq f$.
\item $F_v[t]$, $t\geq 0$, denotes the set of faulty nodes whose round $t$ execution is verified by at least one fault-free node, as per Definition \ref{d_verified}.
\item $\overline{F_v}[t]=F-F_v[t]$, $t\geq 0$.
\item $\alpha=1-\frac{1}{n}$.
\item 
We use boldface upper case letters to denote matrices, rows of matrices, and their elements. For instance, $\bfA$ denotes a matrix, $\bfA_i$ denotes the $i$-th row of matrix $\bfA$, and $\bfA_{ij}$ denotes the element at the intersection of the $i$-th row and the $j$-th column of matrix $\bfA$.

\end{itemize}

\section{$H_l(h_1, h_2,\cdots, h_\nu;~c_1, c_2,\cdots, c_\nu)$ is Convex}
\label{a:linear_convex}

\begin{claim}
$H_l(h_1, h_2,\cdots, h_\nu;~c_1, c_2,\cdots, c_\nu)$ defined in Definition \ref{def:linear_hull}
is convex.
\end{claim}
\begin{proof}

The proof is straightforward.

Let
\[ h_L := H_l(h_1, h_2,\cdots, h_\nu;~c_1, c_2,\cdots, c_\nu)\]
and
\[Q := \{ i ~|~ c_i\neq 0, ~ 1\leq i\leq \nu \}.\]
Given any two points $x, y$ in $h_L$, by Definition \ref{def:linear_hull}, we have
\begin{equation}
\label{eq:x2}
x = \sum_{i \in Q} c_i p_{(i, x)}~~\text{for some}~p_{(i,x)} \in h_i,~~i \in Q
\end{equation}

and 

\begin{equation}
\label{eq:y2}
y = \sum_{i \in Q} c_i p_{(i,y)}~~\text{for some}~p_{(i,y)} \in h_i,~~i\in Q
\end{equation}

Now, we show that any convex combination of $x$ and $y$ is also in $h_L$. Consider a point $z$ such that

\begin{equation}
\label{eq:z2}
z = \theta x + (1 - \theta) y~~~\text{where}~0 \leq \theta \leq 1
\end{equation}

Substituting (\ref{eq:x2}) and (\ref{eq:y2}) into (\ref{eq:z2}), we have

\begin{align}
z &= \theta \sum_{i\in Q}~c_i p_{(i,x)} + (1-\theta) \sum_{i \in Q}~c_i p_{(i,y)}\nonumber\\
&= \sum_{i\in Q}~c_i\left(\theta p_{(i,x)} + (1-\theta) p_{(i,y)}\right) \label{eq:p_iz2}
\end{align}

Define $p_{(i,z)} = \theta p_{(i,x)} + (1-\theta) p_{(i,y)}$ for all $i\in Q$. Since $h_i$ is convex, and $p_{(i,z)}$ is a convex combination of $p_{(i,x)}$ and $p_{(i,y)}$, $p_{(i,z)}$ is also in $h_i$. Substituting the definition of $p_{(i,z)}$ in (\ref{eq:p_iz2}), we have

\begin{align*}
z &=  \sum_{i \in Q}~~c_i~p_{(i,z)}~~\text{where}~p_{(i,z)} \in h_i,~~i\in Q
\end{align*}

Hence, by Definition \ref{def:linear_hull}, $z$ is also in $h_L$. Therefore, $h_L$ is convex.

\end{proof}

\section{Claim \ref{c_common_verify}}
\label{a_claim_common_verify}

\begin{claim}
\label{c_common_verify}
Consider fault-free nodes $i, j\in V-F$.
For $t\geq 0$, if $(h,k,t-1)\in \R_i[t]$ at some point of time, then eventually $(h,k,t-1)\in \R_j[t]$.
\end{claim}

\begin{proof}
The proof is by induction.

{\em Induction basis:}
For round $t=0$,
node $i$ adds $(h,k,-1)$ to $\R_i[0]$ 
whenever it reliably receives message $(h,k,0)$. (For round 0 messages, $h$ is just
a single point.)
Then by {\em Global Liveness} property of reliable broadcast,
node $j$ will eventually reliably receive the same message, and add $(h,k,-1)$ to $\R_j[0]$.

{\em Induction:}
Consider round $t\geq 1$.
Assume that the statement of the lemma holds true through
rounds $t-1$. Therefore, if $(h,k, t-2)\in \R_i[t-1]$ at some point of time, then eventually
$(h,k, t-2)\in \R_j[t-1]$.

Now we will prove that the lemma holds for round $t$.
Suppose that at some time
$\mu$, $(h,k,t-1)\in \R_i[t]$.
Thus, node $i$ must have reliably received (in line 9 of round $t$)
a message of the form $((h,\VV),k,t)$ such that the following
conditions are true at some real time time $\mu$:
\begin{itemize}
\item Condition 1: $\VV\subseteq \R_i[t-1]$  (due to line 10, and the fact that $\R_i[t-1]$ can
	only grow with time)
\item Condition 2: $|\VV| \geq n-f$, $(*,k,t-2)\in \VV$ and $h=H(\VV,t-1)$ (due to line 11)
\end{itemize}
The correctness of the lemma through round $t-1$ implies 
that eventually each element of $\R_i[t-1]$ will be included
in $\R_j[t-1]$. Thus, because $\VV\subseteq \R_i[t-1]$ at time $\mu$, eventually $\VV\subseteq \R_j[t-1]$.
Also, the {\em Global Liveness} property 
implies that eventually node $j$ will reliably receive the
same message $((h,\VV),k,t)$ that was reliably received by node $i$; therefore,
as in Condition 2 above, node $j$ will also find that
$|\VV| \geq n-f$, $(*,k,t-2)\in \VV$ and $h=H(\VV,t-1)$.
Therefore, by lines 10-12, it follows that eventually
$(h,k,t-1)\in \R_j[t]$.
\end{proof}

\section{Proof of Lemma \ref{l_progress}}
\label{a_lemma_progress}

\noindent
{\bf Lemma \ref{l_progress}:}
{\em
 If all the fault-free nodes progress to the start of round $t$, $t \geq 0$, then all the
fault-free nodes will eventually progress to the start of round $t+1$. \\
}

\begin{proof}
The proof is by induction.
By assumption, all nodes begin round 0 eventually, and perform reliable broadcast
of their input (line 1).
By {\em Fault-Free Liveness} property of {\em reliable broadcast},
each fault-free node $i$ will
eventually reliably receive messages from all the $n-f$ fault-free nodes.
All the messages reliably received in round 0 result
in addition of an element to the set $\R_i[0]$ at fault-free node $i$ (line 3);
therefore, $\R_i[0]$ will eventually be of size at least $n-f$.
It follows that each fault-free node will eventually complete round 0,
and proceed to round 1 (lines 4-7).

Now we assume that all the fault-free nodes have progressed to the start of round $t$,
where $t\geq 1$, and prove that all the fault-free nodes will eventually progress
to the start of round $t+1$.

Consider fault-free nodes $i,j\in V-F$.
In line 8 of round $t$, fault-free node $j$ performs reliable broadcast of $((h_j[t-1],\R^c_j[t-1]),j,t)$.
By {\em Fault-free Liveness} of reliable broadcast, fault-free node $i$ will eventually
reliably receive message $((h_j[t-1],\R^c_j[t-1]),j,t)$ 
from fault-free node $j$.
By Claim \ref{c_common_verify}, eventually 
$\R^c_j[t-1] \subseteq \R_i[t-1]$; therefore, node $i$ will progress past line 10
in the handler for message $((h_j[t-1],\R^c_j[t-1]),j,t)$.
Moreover, since node $j$ is fault-free, it follows the algorithm
specification correctly. Therefore, the checks on line 11 in the handler
at node $i$ for message $((h_j[t-1],\R^c_j[t-1]),j,t)$ will all be correct.
Therefore, by lines 11-12, node $i$ will eventually include $(h_j[t-1],j,t-1)$ in $\R_i[t]$.
Since the above argument holds for all fault-free nodes $i,j$, it implies that
each fault-free node $i$ eventually adds $(h_j[t-1],j,t-1)$ to $\R_i[t]$, for each
the fault-free node $j$ (including $j=i$). Therefore, at each fault-free node $i$,
eventually,
$|\R_i[t]|\geq n-f$,
and $(h_i[t-1],i,t-1)\in \R_i[t]$ (because the previous statement holds for $j=i$ too), thus satisfying both the conditions
at line 13.
Thus, each fault-free node $i$ will eventually proceed to round $t+1$ (lines 13-16).

\end{proof}

\section{Claims \ref{c_ver} and  \ref{claim:j_verified}}
\label{a_claims}

\begin{claim}
\label{c_ver}
If faulty node $i$'s round $t$ execution is verified by a fault-free node $j$, then
the following statements hold:

(i) For $t\geq 0$, $\R^c_i[t]\geq n-f$ and $h_i[t]=H(\R^c_i[t],t)$,

(ii) For $t\geq 0$, eventually $\R^c_i[t]\subseteq \R_j[t]$, and

(iii) For $t\geq 1$, node $i$'s round $t-1$ execution is also verified by node $j$.
\end{claim}
\begin{proof}
Let $t\geq 0$.
Suppose that node $i$'s round $t$ execution is verified by a fault-free node $j$.
In this case, we can use definitions (\ref{e_faulty_h}) and (\ref{e_faulty_V})
of $h_i[t]$ and $\R^c_i[t]$.
Definition \ref{d_verified} implies that node $j$ eventually 
reliably receives message $((h_i[t],\R^c_i[t]),i,t+1)$ from node $i$,
and subsequently adds (at line 12 of its round $t+1$) $(h_i[t],i,t)$ to $\R_j[t+1]$.
This implies that this message satisfies the checks done by node $j$ at line 11:
Specifially,
(a) $|\R^c_i[t]|\geq n-f$ and $h_i[t]=H(\R^c_i[t],t)$, and
(b) $(*,i,t-1)\in\R^c_i[t]$.
Also, by the time node $j$ adds $(h_i[t],i,t+1)$ to $\R_j[t+1]$, the condition
on line 10 also holds: specifically, $\R^c_i[t]\subseteq \R_j[t]$, proving
claim $(ii)$ stated above.
Also, (a) above proves claim $(i)$.

$(*,i,t-1)\in\R^c_i[t]$ and $\R^c_i[t]\subseteq \R_j[t]$
together imply 
that eventually $(*,i,t-1)\in \R_j[t]$.

Suppose that $t\geq 1$. Then
the above observation that eventually $(*,i,t-1)\in \R_j[t]$, and 
Definition \ref{d_verified}, imply that round $t-1$ execution of node $i$
is verified by node $j$. This proves $(iii)$.
\end{proof}

~

\begin{claim}
\label{claim:j_verified}
If faulty node $i$'s round $t$ execution is verified by a fault-free node $j$, $t \geq 0$,
then for all $r$ such that $0\leq r \leq t$, node $i$'s round $r$ execution is verified by node $j$.
\end{claim}

\begin{proof}
The claim is trivially true for $t=0$.
The proof of the claim for $t>0$ follows by repeated application of Claim \ref{c_ver}(iii) above.
\end{proof}

~


\section{Proof of Lemma \ref{lemma:J_in_H0}}
\label{a_lemma_J}

The proof of Lemma \ref{lemma:J_in_H0} uses the following theorem by Tverberg \cite{tverberg}:
\begin{theorem}
\label{thm:tverberg}
(Tverberg's Theorem \cite{tverberg}) For any integer $f \geq 1$, for every multiset $Y$ containing at least $(d+1)f+1$ points in a $d$-dimensional space, there exists a partition
$Y_1, .., Y_{f+1}$ of $Y$ into $f+1$ non-empty multisets such that $\cap_{l=1}^{f+1} \HH(Y_l) \neq \emptyset$.
\end{theorem}

~

\noindent
Now we prove Lemma \ref{lemma:J_in_H0}. \\

\noindent{\bf Lemma \ref{lemma:J_in_H0}:}
{\em
 For each node $i\in V-\overline{F_v}[0]$, the polytope $h_i[0]$ is non-empty. \\
}

\begin{proof}
%
Note that $V-\overline{F_v}[0]=(V-F)\cup F_v[0]$.
\begin{itemize}
\item
For a fault-free node $i\in V-F$, since it behaves correctly,
$|\R^c_i[0]|\geq n-f$ and
$h_i[0]=H(\R^c_i[0],0)$, due to lines 4-7.
\item
For faulty node $i\in F_v[0]$ as well, by Claim \ref{c_ver}(i)
in Appendix \ref{a_claims},
$|\R^c_i[0]|\geq n-f$ and
$h_i[0]=H(\R^c_i[0],0)$.
\end{itemize}
Thus, for each $i\in V-\overline{F_v}[0]$,
$|\R^c_i[0]|\geq n-f$ and
$h_i[0]=H(\R^c_i[0],0)$.

Consider any $i\in V-\overline{F_v}[0]$.
Consider the computation of polytope $h_i[0]$ as $H(\R^c_i[0],0)$.
By step 1 of function $H$ in Section \ref{s_ops}, $|X|=|\R^c_i[0]|\geq n-f$.
Recall that, due to the lower bound on $n$ discussed in Section \ref{s_intro},
we assume $n \geq (d+2)f+1$.
Thus, in function $H$, $|X| \geq n-f \geq (d+1)f+1$. By Theorem \ref{thm:tverberg}, there exists a partition $X_1, X_2, \cdots, X_{f+1}$ of $X$ into multisets $X_j$ such that $\cap_{j=1}^{f+1} \HH(X_j) \neq \emptyset$. Let us define
\begin{equation}
\label{eq:J}
J = \cap_{i=1}^{f+1} \HH(X_j)
\end{equation}
Thus, $J$ is non-empty.
In step 2 of function $H$ (for $t=0$), because $|X|\geq n-f$,
each multiset $C$ used in the computation of function $H$ is of size at least $n-2f$.
Thus, each $C$ excludes only $f$ elements of $X$, whereas there are $f+1$ multisets
in the above partition of $X$.
Therefore, each set $C$ in step 2 of function $H$ will fully contain at least one multiset $X_j$
from the partition. Therefore, $\HH(C)$
will contain $J$. Since this holds true for all $C$'s, $J$ is contained in
the convex polytope computed by function $H$.
Since $J$ is non-empty, $h_i[0]=H(\R^c_i[0],0)$ is non-empty.

\end{proof}

\section{Claim \ref{claim:notFv_not_verified}}
\label{app_s_claim:notFv_not_verified}

\begin{claim}
\label{claim:notFv_not_verified}
For $t\geq 0$,
if $b \in \overline{F_v}[t]$, then for all $i\in V-\overline{F_v}[t+1]$,
$(*,b,t)\not\in\R^c_i[t+1]$.
\end{claim}

\begin{proof}
Consider faulty node $b\in \overline{F_v}[t]$.
Note that $V-\overline{F_v}[t+1]=(V-F)\cup F_v[t+1]$.

\begin{itemize}
\item Consider a fault-free node $i\in V-F$.
Since $b\in \overline{F_v}[t]$, node $b$'s round $t$ execution is {\em not} verified
by {\em any} fault-free node. Therefore, by Definition \ref{d_verified},
for fault-free node $i\in V-F$, {\bf at all times},
$(*,b,t)\not\in \R_i[t+1]$.
Therefore, by line 14, $(*,b,t)\not\in \R^c_i[t+1]$.

\item
Consider a faulty node $i\in F_v[t+1]$.
In this case,
the proof is by contradiction. In particular, for some $h$, assume that $(h,b,t) \in \R^c_i[t+1]$.
Since $i\in F_v[t+1]$, there exists a fault-free node $j$ that verifies the
round $t+1$ execution of node $i$. Therefore, by Claim \ref{c_ver}(ii) in
Appendix \ref{a_claims},
eventually $\R^c_i[t+1]\subseteq \R_j[t+1]$.  This observation, along with the above assumption
that $(h,b,t) \in \R^c_i[t+1]$, implies that eventually $(h,b,t)\in \R_j[t+1]$.
Since node $j$ is fault-free, Definition \ref{d_verified} implies that execution of
node $b$ in round $t$ is verified, and hence $b\in F_v[t]$. This is a contradiction.
Therefore, $(*,b,t)\not\in \R^c_i[t+1]$.
\end{itemize}
\end{proof}

\section{Proof of Lemma \ref{lemma:always_notFv_not_verified}}
\label{a_always}

\noindent{\bf Lemma \ref{lemma:always_notFv_not_verified}:}
{\em
 For $r\geq 0$,
if $b \in \overline{F_v}[r]$, then for all $\tau\geq r$,
\begin{itemize}
\item $b \in \overline{F_v}[\tau]$, and 
\item for all $i\in V-\overline{F_v}[\tau+1]$, $(*,b,\tau)\not\in\R^c_i[\tau+1]$.
\end{itemize}
}

\begin{proof}

Recall that $F_v[r]\subseteq F$, and $\overline{F_v}[r] = F-F_v[r]$.

For $r \geq 0$, consider a faulty node $b \in \overline{F_v}[r]$.
Thus, $b\in F$.

We first prove that
$b \in \overline{F_v}[\tau]$, for $\tau\geq r$.
This is trivially true for $\tau=r$. So we only need to prove this
for $\tau>r$.
The proof is by contradiction.

Suppose that
there exists $\tau>r$ such that $b \not\in \overline{F_v}[\tau]$.
Thus, $b\in F_v[\tau]$.
The definition of $\overline{F_v}[\tau]$ implies that 
node $b$'s round $\tau$ execution is verified by some fault-free node $j$.
Then Claim \ref{claim:j_verified} implies that
node $b$'s round $r$ execution is verified
by node $j$.
Hence by the definition of $F_v[r]$, $b\in F_v[r]$. This is a contradiction. 
This proves that $b\in\overline{F_v}[\tau]$.

Now, since $b\in\overline{F_v}[\tau]$, by
Claim \ref{claim:notFv_not_verified},
for all $i\in V-\overline{F_v}[\tau+1]$, $(*,b,\tau)\not\in\R^c_i[\tau+1]$.

\end{proof}

\section{Claims \ref{c_single}, \ref{c_common} and \ref{claim:at_least_one_verified}}
\label{a_claim:one}

\begin{claim}
\label{c_single}
For $t \geq 0$, a fault-free node $i$ adds at most one message from node $j$ to $\R_i[t]$, even if $j$ is faulty.
\end{claim}

\begin{proof}
As stated in the properties of the primitive in Section \ref{s_ops}, each fault-free
node $i$ will reliably receive at most one message of the form $(*,j,t)$ from node $j$. Since $\R_i[t]$ only contains tuples corresponding to reliably received
messages, the claim follows.
\end{proof}

\begin{claim}
\label{c_common}
For $t\geq 1$, consider nodes $i,j\in V-\overline{F_v}[t]$.
If $(h,k,t)\in \R^c_i[t]$ and $(h',k,t)\in\R^c_j[t]$, then
$h=h'$.
\end{claim}
\begin{proof}
We consider four cases:
\begin{itemize}
\item $i,j\in V-F$:
In this case, due to {\em Global Uniqueness} property of reliable broadcast,
nodes $i$ and $j$ cannot reliably receive different round $t$ messages from the same
node. Hence the claim follows.

\item $i\in V-F$ and $j\in F_v[t]$:
Suppose that fault-free node $p$ verifies round $t$ execution of node $j$. Then
by Claim \ref{c_ver}(ii), eventually $\R^c_j[t]\subseteq \R_p[t]$.
Since nodes $i$ and $p$ are both fault-free. Therefore, similar to the previous
case, due to the {\em Global Uniqueness} property, nodes $i$ and $p$
cannot reliably receive distinct round $t$ messages.
Thus, if $(h,k,t)\in \R^c_i[t]$ and $(h',k,t)\in\R^c_j[t]\subseteq \R_p[t]$,
then $h=h'$.

\item $j\in V-F$ and $i\in F_v[t]$: This case is similar to the previous case.
\item $i,j\in F_v[t]$:
In this case, there exist fault-free nodes $k_i$ and $k_j$
that verify round $t$ execution of nodes $i$ and $j$, respectively.
Thus, by Claim \ref{c_ver}(ii), eventually $(h,i,t)\in \R^c_i[t]\subseteq \R_{k_i}[t]$ and $(h',i,t)\in\R^c_j[t]\subseteq \R_{k_j}[t]$.
Since $k_i,k_j$ are fault-free, {\em Global Uniqueness} implies that $h=h'$.
\end{itemize}
\end{proof}

\begin{claim}
\label{claim:at_least_one_verified}
For $t\geq 1$,
consider nodes $i, j \in V - \overline{F_v}[t]$. There exists a fault-free node $g\in V-F$ such that $(h_g[t-1],g,t-1) \in \R^c_i[t]\cap\R^c_j[t]$.
\end{claim}

\begin{proof}
For any fault-free node, say $p$, due to the conditions checked on line 13,
$|\R^c_p[t]|\geq n-f$.
For a node $k\in F_v[t]$, recall that $h_k[t]$ and $\R^c_k[t]$ are defined in
(\ref{e_faulty_h}) and (\ref{e_faulty_V}).
Thus, by Definition \ref{d_verified}, there exists some fault-free node, say $q$, that reliably receives
message $((h_k[t],\R^c_k[t]),k,t+1)$ from node $k$ in round $t+1$, and after performing checks
on line 13, adds $(h_k[t],k,t)$ to $\R^c_q[t+1]$. The checks on line 13, performed by 
fault-free node $q$, ensure that $|\R^c_k[t]|\geq n-f$.

Above argument implies that for the nodes $i,j\in V-\overline{F_v}[t]$,
$\R^c_i[t]$ and $\R^c_j[t]$ both contain at least $n-f$ messages.
Therefore, by Claims \ref{c_single} and \ref{c_common}, there will be at least $n-2f \geq df+1\geq f+1$ elements in $\R^c_i[t]\cap\R^c_j[t]$.
Since $f$ is the upper bound on the number of faulty nodes,
at least one element in $\R^c_i[t]\cap\R^c_j[t]$ corresponds to
a fault-free node, say node $g\in V-F$. That is, there exists $g\in V-F$ such that
$(h_g[t-1],g,t-1) \in \R^c_i[t] \cap \R^c_j[t]$.

\end{proof}

\section{Proof of Lemma \ref{lemma:transition_matrix}}
\label{a_t_matrix}

\noindent{\bf Lemma 
\ref{lemma:transition_matrix}:}
{\em
 For $t\geq 1$,
transition matrix $\bfM[t]$ constructed using the above procedure satisfies the following conditions.
\begin{itemize}
\item For $i,j \in V$, there exists a fault-free node $g(i,j)$ such that $\bfM_{ig(i,j)}[t] \geq \frac{1}{n}$.

\item $\bfM[t]$ is a row stochastic matrix, and $\lambda(\bfM[t])\leq 1-\frac{1}{n}$.


\end{itemize}
}

\begin{proof}

\begin{itemize}
\item To prove the first claim in the lemmas, we consider four cases for node pairs $i,j$.

(i) $i,j\in V-\overline{F_v}[t]$:

 By Claim \ref{claim:at_least_one_verified},
there exists a node $g(i,j)$ such that
$(h_{g(i,j)}[t-1],g(i,j),t-1) \in \R^c_i[t]\cap \R^c_j[t]$.
By (\ref{eq:matrix_i}) in the procedure to construct $\bfM[t]$, $\matrixm_{ig(i,j)}[t]=\frac{1}{|\R^c_i[t]|}\geq \frac{1}{n}$
and 
$\matrixm_{jg(i,j)}[t]=\frac{1}{|\R^c_j[t]|}\geq\frac{1}{n}$.

(ii) $i\in\overline{F_v}[t]$ and $j\in V-\overline{F_v}[t]$:

$|\R^c_j[t]|\geq n-f$ elements of
$\matrixm_j[t]$ are equal to $\frac{1}{|\R^c_j[t]|}\geq \frac{1}{n}$.
Since $n-f\geq (d+1)f+1\geq 2f+1$, there exists a fault-free node $g(i,j)$ such that
$\matrixm_{jg(i,j)}\geq \frac{1}{n}$.
By (\ref{e_fv}), all elements of $\matrixm_i[t]$, including $\matrixm_{ig(i,j)}[t]=\frac{1}{n}$.

(iii) $j\in\overline{F_v}[t]$ and $i\in V-\overline{F_v}[t]$:
Similar to case (ii).

(iv) $i,j\in\overline{F_v}[t]$: 

By (\ref{e_fv}) in the procedure to construct $\bfM[t]$, all $n$ elements in $\matrixm_i[t]$ and $\matrixm_j[t]$
both equal $\frac{1}{n}$. Choose a fault-free node as node $g(i,j)$.
Then
$\matrixm_{ig(i,j)}[t]=
\matrixm_{ig(i,j)}[t]=\frac{1}{n}$.

\item Observe that, by construction, for each $i\in V$, the row vector $\bfM_i[t]$ 
is stochastic. Thus, $\bfM[t]$ is row stochastic.
Also, due to the claim proved in the previous item, and
Claim \ref{c_lambda_bound}, $\lambda(\matrixm[t])\leq 1-\frac{1}{n}<1$.

\end{itemize}
\end{proof}

\section{Proof of Lemma \ref{lemma:valid_initial_hull}}
\label{app_s_lemma:valid_initial_hull}

\noindent{\bf Lemma \ref{lemma:valid_initial_hull}:} {\em
 $h_i[0]$ for each node $i \in \sv - \overline{F_v}[0]$ is valid. \\
}

\begin{proof}
 Consider two cases:

\begin{itemize}
\item $i \in \sv - F$:
Due to Claim \ref{c_single} in Appendix \ref{a_claim:one} and the fact that $i$ is fault-free $\R^c_i[0]$ contains at most $f$ elements
corresponding to
faulty nodes. Recall that $h_i[0]$ is obtained using function $H(\R^c_i[0],0)$.
Then, due to the {\em Fault-Free Integrity} property of reliable broadcast,
and the definition of function $H$, at least one set $C$ used
in item 2 of function $H$ will contain only the inputs of fault-free nodes. 
Therefore, $h_i[0]$ is in the convex hull of the inputs at fault-free nodes.
That is, $h_i[0]$ is valid.

\item $i \in F_v[0]$:
Suppose that round $0$ execution of node $i$ is verified by fault-free node $j$.
By Claim \ref{c_ver} in Appendix \ref{a_claims}, $h_i[0]=H(\R^c_i[0],0)$, and eventually
$\R^c_i[0]\subseteq \R_j[0]$.
Thus, eventually, $H(\R^c_i[0],0)\subseteq H(\R_j[0],0)$.
Since $j$ is fault-free, and $|\R_j[0]|\geq n-f$, by an argument similar to the previous item,
$H(\R_j[0],0)$ is valid. This implies that $h_i[0]=H(\R^c_i[0],0)$ is also valid.

\end{itemize}

\end{proof}

\section{Proof of Lemma \ref{lemma:linear_valid}}
\label{app_s_lemma:linear_valid}

The proof is straightforward, but included here for completeness.

\noindent{\bf Lemma \ref{lemma:linear_valid}:} {\em
 Suppose non-empty convex polytopes $h_1, h_2, \cdots, h_k$ are all valid. Consider $k$ constants $c_1, c_2, \cdots, c_k$ such that $0 \leq c_i \leq 1$ and $\sum_{i = 1}^k c_i = 1$.
Then the linear combination of these convex polytopes,
$H_l(h_1, h_2, \cdots, h_k; c_1, c_2, \cdots, c_k)$, is valid. \\
}

\begin{proof}

Observe that the points in $H_l(h_1,\cdots,h_k;c_1,\cdots,c_k)$ are convex
combinations of the points in $h_1,\cdots,h_k$, because $\sum_{i=1}^k c_i = 1$ and $0\leq c_i\leq 1$,
for $1\leq i\leq k$.
Let $G$ be the set of input vectors at the fault-free nodes
in $V-F$. Then, $\HH(G)$ is the convex hull of the inputs at the fault-free nodes.
Since $h_i$, $1\leq i\leq k$, is valid, each point $p\in h_i$ is in $\HH(G)$.
Since $\HH(G)$ is a convex polytope, it follows that any convex combination
of the points in $h_1,\cdots,h_k$ is also in $\HH(G)$.

\end{proof}

\section{Algebraic Manipulation in the Proof of Theorem \ref{thm:correctness}}
\label{a_hausdorff}

\begin{align}
d(p_i^*, p_j^*) &=  \sqrt{\sum_{l=1}^d (p_i^*(l) - p_j^*(l))^2} \nonumber\\
&= \sqrt{\sum_{l=1}^d \left(\sum_{k\in V-\overline{F_v}[0]} \bfM^*_{ik} p_k(l) - \sum_{k\in V-\overline{F_v}[0]} \bfM^*_{jk} p_k(l)\right)^2}\nonumber ~~~~ \text{by (\ref{pi}) and (\ref{pj}})\\
&= \sqrt{\sum_{l=1}^d \left(\sum_{k\in V-\overline{F_v}[0]} (\bfM^*_{ik}-\bfM^*_{jk}) p_k(l)\right)^2}\nonumber\\
&\leq \sqrt{\sum_{l=1}^d \left[\alpha^{2t}\left( \sum_{k\in V-\overline{F_v}[0]} \| p_k(l) \|\right)^2\right]} 
~~~~~~~ \mbox{~~~~ by (\ref{eq:delta1})}\nonumber\\
&= \alpha^t \sqrt{\sum_{l=1}^d \left(\sum_{k\in V-\overline{F_v}[0]} \| p_k(l) \| \right)^2}\label{eq:d_pipj}
\end{align}

\section{Proof of Lemma \ref{lemma:opt}}
\label{a_opt}

We first prove a claim that will be used in the proof of Lemma \ref{lemma:opt}.

\begin{claim}
\label{claim:M*}
For $t \geq 1$, define $\bfM'[t] = \Pi_{\tau=1}^t \bfM[\tau]$. Then, for all nodes $j \in V - \overline{F_v}[t]$, and $k \in \overline{F_v}[0]$, $\bfM'_{jk}[t] = 0$.
\end{claim}

\begin{proof}
The proof is by induction on $t$.

{\em Induction Basis}: Consider the case when $t = 1$. Recall that $V - \overline{F_v}[1] = (V-F) \cup F_v[1]$. Consider
any $j\in V - \overline{F_v}[1]$, and $k \in \overline{F_v}[0]$. Then by 
Lemma \ref{lemma:always_notFv_not_verified}, $(*,k,0) \not\in \R^c_j[1]$.
Then, due to (\ref{eq:matrix_i-2}),
$\bfM_{jk}[1]=0$, and hence $\bfM'_{jk}[1]=\bfM_{jk}[1]=0$.

{\em Induction}: Consider $t \geq 2$. Assume that the claim holds true through  $t - 1$. 
Then, $\bfM'_{jk}[t-1]=0$ for all $j \in V - \overline{F_v}[t-1]$ and $k \in \overline{F_v}[0]$.
Recall that $\bfM'[t-1] = \Pi_{\tau=1}^{t-1} \bfM[\tau]$.

Now, we will prove that the claim holds true for $t$.
Consider $j \in V - \overline{F_v}[t]$
and $k\in \overline{F_v}[0]$.
Note that $\bfM'[t] = \Pi_{\tau=1}^{t} \bfM[\tau] = \bfM[t] \Pi_{\tau=1}^{t-1} \bfM[\tau] = \bfM[t] \bfM'[t-1]$.
Thus, $\bfM'_{jk}[t]$ can be non-zero only if there exists a $q \in V$ such that $\bfM_{jq}[t]$ and $\bfM'_{qk}[t-1]$
are both non-zero.

For any $q\in \overline{F_v}[t-1]$, by
Lemma (\ref{lemma:always_notFv_not_verified}), $(*,q,t-1) \not\in \R^c_j[t]$.
Then, due to (\ref{eq:matrix_i-2}), $\bfM_{jq}[t] = 0$ for all $q \in \overline{F_v}[t-1]$.
 Additionally, by the induction hypothesis,
 for all $q \in V - \overline{F_v}[t-1]$ and $k \in \overline{F_v}[0]$,
 $\bfM'_{qk}[t-1] = 0$.
Thus, these two observations together imply that there does not exist any $q \in V$ such that
$\bfM_{jq}[t]$ and $\bfM'_{qk}[t-1]$ are both non-zero.
Hence, $\bfM'_{jk}[t]=0$.
\end{proof}

~

~

Recall from (\ref{eq:I}) that
\[
I = \cap_{D \subset G, |D| = n - 2f-\phi}\, \HH(D)
\]
where
$\phi=|F|$,
and $G$ is the set of inputs at the $n-\phi$ fault-free nodes. \\

\noindent{\bf Lemma \ref{lemma:opt}:}{\em For all $i\in V-F$ and $t\geq 0$,
  $I\subseteq h_i[t]$. \\
}

\begin{proof}
We first prove that the convex polytope $I$ is contained in $h_i[0]$ for all $i \in \sv - \overline{F_v}[0]$. Note that $V-\overline{F_v}[0]=(V-F)\cup F_v[0]$.\\

 Consider two cases:

\begin{itemize}
\item $i \in \sv-F$:

Consider the computation of $h_i[0]$ at fault-free node $i$ using function
$H(\R^c_i[0],0)$ on line 6.
Observe from the definition of function $H$ in Section \ref{s_ops} that
\[ X:=\{h~|~ (h,j,-1)\in \R^c_i[0],~ j\in V\}\]
and 
\[ h_i[0]~:= ~ \cap_{\,C \subset X, |C| = |X| - f}~~\HH(C).\]
Since $|X|=|\R^c_i[0]| \geq n-f$ and $\R^c_i[0]$ contains tuples corresponding
to at most $|F|=\phi$ nodes, 
$|X \cap G| \geq n - f-\phi$. Therefore, every multiset $C$ in the computation
of $h_i[0]$ contains inputs of at least $n-2f-\phi$ fault-free nodes.
Thus, $h_i[0]$ contains $I$.

\item $i \in F_v[0]$:

Suppose that round $0$ execution of node $i$ is verified by a fault-free node $j$.
By Claim \ref{c_ver}, eventually
$\R^c_i[0]\subseteq \R_j[0]$, and $h_i[0]=H(\R^c_i[0],0)$.
Since node $j$ is fault-free, at most $\phi$ elements in $\R_j[t]$ correspond
to 
faulty nodes. Therefore, at most $\phi$ elements in $\R^c_i[t]$ correspond to faulty
nodes, and at least $|\R^c_i[t]|-\phi \geq n-f-\phi$ correspond to fault-free nodes.
Thus, similar to the previous case,
$|X \cap G| \geq n - f-\phi$. Therefore, every multiset $C$ in the computation
of $h_i[0]$ contains inputs of at least $n-2f-\phi$ fault-free nodes.
Thus, $h_i[0]$ contains $I$.

\end{itemize}

\noindent
Now we make several observations for each fault-free node $i\in V-F$:
\begin{itemize}
\item As shown above, $I \in h_j[0]$ for all $j\in V-\overline{F_v}[0]$. 
\item
From (\ref{e_Mstar}), for $t\geq 1$,
\[
\vectorv[t]=\bfM^* \vectorv[0]
\]
where $\vectorv_j[0]=h_j[0]$ for $j\in V-\overline{F_v}[0]$.
\item By Theorem \ref{t_M},  $\vectorv_i[t] = h_i[t]$.
\item
 Observe that $\bfM^*$ equals $\bfM'[t]$ defined in Claim \ref{claim:M*}.
Thus, due to Claim \ref{claim:M*}, $\bfM^*_{ik}=0$ for $k\in\overline{F_v}[0]$ (i.e,. $k\not\in V-\overline{F_v}[0]$).
\item

$\bfM^*$ is the product of row stochastic matrices; therefore, $\bfM^*$ itself is also row stochastic.
Thus, for fault-free node $i$, $\vectorv_i[t]=h_i[t]$ is obtained
as the product of the $i$-th row of $\bfM^*$, namely $\bfM^*_i$, and $\vectorv[0]$:
this product yields 
a linear combination of the elements of $\vectorv[0]$, where the weights
are non-negative and add to 1 (because $\bfM^*_i$ is a stochastic row vector).

\item
From (\ref{e_r_c}), recall that $\bfM_i^*\vectorv[0]=H_l(\vectorv[0]^T~;~\bfM^*_i)$.
Function $H_l$ ignores the input polytopes for which the corresponding weight is 0.
Finally, from the previous observations, we have that when the weight in $\bfM^*[i]$ 
is non-zero, the corresponding polytope in 
$\vectorv[0]^T$ contains $I$. 
Therefore, the linear combination also contains $I$.
\end{itemize}
Thus,
$I$ is contained in $h_i[t]=\vectorv_i[t]=\bfM^*_i\vectorv[0]$.

\comment{+++++++++++ old proof +++++++++
Now consider any node $i\in V-F$.
Recall from (\ref{e_Mstar}) that, for $t\geq 1$.
\[
\vectorv[t]=\bfM^* \vectorv[0]
\]
where
$\vectorv_j[0]=h_j[0]$ for $j\in V-\overline{F_v}[0]$,
and by Theorem \ref{t_M},  $\vectorv_i[t] = h_i[t]$.
$\bfM^*$ is the product of row stochastic matrices; therefore, $\bfM^*$ itself is also
row stochastic.
Thus, for fault-free node $i$, $\vectorv_i[t]=h_i[t]$ is obtained
as the product of the $i$-th row of $\bfM^*$, namely $\bfM^*_i$, and $\vectorv[0]$:
this product yields 
a linear combination of the elements of $\vectorv[0]$, where the weights
are non-negative and add to 1 (because $\bfM^*_i$ is a stochastic row vector).

Additionally,
by Lemma \ref{lemma:transition_matrix}, if $k\in \overline{F_v}[0]$,
then $k\in \overline{F_v}[\tau]$ for $\tau\geq 0$.
By definition of $\overline{F_v}[\tau]$,
because $k\in\overline{F_v}[\tau]$,
$(*,k,\tau)\not\in \R^c_i[\tau+1]$.
Therefore,
due to (\ref{eq:matrix_i-2}),
$\bfM_{ik}[\tau+1]=0$.
Thus, $\bfM_{ik}[\tau+1]=0$ for all $\tau\geq 0$
(i.e., $\bfM_{ik}[r]=0$ for $r\geq 1$).
Then the definition of $\bfM^*$ implies that for $k\in \overline{F_v}[0]$, $\bfM^*_{ik}=0$.

Finally, from (\ref{e_r_c}), recall that $\bfM_i^*\vectorv[0]=H_l(\vectorv[0]^T~;~\bfM^*_i)$,
and function $H_l$ 
ignores the input polytopes for which
the corresponding weight is 0.

The above observations together imply that, since $I$ is contained
in $h_i[0]$ for all $i\in V-\overline{F_v}[0]$, $I$ is also contained
in $h_i[t]=\vectorv_i[t]=\bfM^*_i\vectorv[0]$.
++++++++++}

\end{proof}

\end{document}